\documentclass[11pt,a4paper]{article}
\usepackage[utf8]{inputenc}
\usepackage{fullpage}
\usepackage{geometry}
\usepackage{amsmath}
\usepackage{amsthm}
\usepackage{amsfonts}
\usepackage{bbm}
\usepackage{color}
\usepackage[ruled,vlined]{algorithm2e}
\usepackage{caption}
\usepackage{graphicx}
\usepackage{subcaption}
\usepackage{mathtools}
\DeclarePairedDelimiter{\ceil}{\lceil}{\rceil}
\DeclarePairedDelimiter\floor{\lfloor}{\rfloor}

\newtheorem{example}{Example} 
\newtheorem{theorem}{Theorem}
 
\newtheorem{proposition}[theorem]{Proposition}

\newtheorem{definition}[theorem]{Definition}

\usepackage[unicode=true,pdfusetitle,
bookmarks=true,bookmarksnumbered=false,bookmarksopen=false,
breaklinks=false,pdfborder={0 0 0},pdfborderstyle={},backref=false,colorlinks=true]{hyperref}
\hypersetup{urlcolor=burntgreen,linkcolor=blue,citecolor=blue}
\usepackage{natbib}

\usepackage{xcolor}
\definecolor{burntorange}{rgb}{0.85, 0.35, 0.1}
\definecolor{charcoal}{rgb}{0.21, 0.27, 0.31}
\definecolor{coolblack}{rgb}{0.0, 0.28, 0.49}
\definecolor{burntgreen}{rgb}{0.05, 0.45, 0.27}
\definecolor{burntblue}{rgb}{0.05, 0.27, 0.8}

\newcommand{\Palg}{\textsc{PolyGS}}
\newcommand{\ON}{\textsc{offerNext}}
\newcommand{\WM}{\textsc{weakestMatch}}
\newcommand{\GSalg}{GS algorithm}
\newcommand{\TO}{\textsc{mapToExtendedMarket}}
\newcommand{\From}{\textsc{mapFromExtendedMarket}}

\title{Two-Sided Matching Markets in the ELLIS 2020 PhD Program}
\author{
Maximilian Mordig$^{1, 2}$
\and Riccardo Della Vecchia$^3$
\and Nicol\`o Cesa-Bianchi$^4$
\and Bernhard Sch\"olkopf$^{1, 2}$
}
\date{%
$^1$ Max-Planck Institute for Intelligent Systems Tübingen\\%
$^2$ ETH Zürich\\%
$^3$ Artificial Intelligence Lab, Bocconi University, Milano, Italy\\%
$^4$ Dept. of Computer Science \& DSRC, University of Milan, Italy\\[2ex]%
\today
}

\begin{document}

\maketitle

\begin{abstract}
The ELLIS PhD program is a European initiative that supports excellent young researchers by connecting them to leading researchers in AI. In particular, PhD students are supervised by two advisors from different countries: an advisor and a co-advisor.
In this work we summarize the procedure that, in its final step, matches students to advisors in the ELLIS 2020 PhD program. 
The steps of the procedure are based on the extensive literature of two-sided matching markets and the college admissions problem \citep{knuth1997stable,gale1962college,roth1992two}.
We introduce \Palg{}, an algorithm for the case of two-sided markets with quotas on both sides (also known as many-to-many markets) which we use throughout the selection procedure of pre-screening, interview matching and final matching with advisors.
The algorithm returns a stable matching in the sense that no unmatched persons prefer to be matched together rather than with their current partners (given their indicated preferences). 
\cite{roth1984evolution} gives evidence that only stable matchings are likely to be adhered to over time.
Additionally, the matching is student-optimal.
Preferences are constructed based on the rankings each side gives to the other side and the overlaps of research fields.
We present and discuss the matchings that the algorithm produces in the ELLIS 2020 PhD program.

\end{abstract}

\section{Introduction}
The ELLIS PhD program is a European initiative that supports excellent young researchers by connecting them to leading researchers in AI. In particular, PhD students are supervised by two advisors from different countries.
This document summarizes the procedure that, in its final step, matches students to advisors in the ELLIS 2020 PhD program. 
After an overview of the selection procedure, we present the theory of two-sided matching markets which is used first to match the students to the evaluators for an initial pre-screening, then to match candidates to advisors in the interviews, and finally is responsible for matching students to advisors in ELLIS 2020. 
We propose an algorithm that matches students to advisors. This matching is a suggestion and may be inspected and corrected manually to satisfy additional criteria.
The co-advisor must be manually matched to a student and must be from a different country than the main advisor in ELLIS.
Briefly, each student ranks research fields and/or professors they would like to work with and vice versa for professors. This happens in the following order.
First, when students apply, they are pre-screened using just the preferences of students over professors and the overlaps between fields of students and advisors. 
After this, advisors also rank the acceptable candidates and the interviews are set. 
After the interviews, both students and professors can consequentially update their preferences and the final matching takes place, matching each student to an advisor. The co-advisor is added separately, taking care of possible constraints imposed by the program. 

In \citep{mordig2021multisided}, we address the case when both advisors and co-advisors should be matched algorithmically to a student without human aid. 
We assume that all preferences are available to us. When the acquisition of preferences is too expensive, \cite{charlin2013toronto} uses machine-learning approaches to interpolate sparse preference data to other persons.
Furthemore, our procedure relies heavily upon the seminal work on two-sided markets by \cite{gale1962college} which had high practical impact since its introduction, leading to improved matching systems for high-school admissions and labor markets \citep{roth1984evolution}, house allocations with existing tenants \citep{abdulkadirouglu1999house}, content delivery networks \citep{maggs2015algorithmic}, and kidney exchanges \citep{roth2005pairwise}.

We introduce the different phases of the selection procedure in a broad and schematic way in Section~\ref{sec:phases}. In Section~\ref{sec:polygamous} we present the algorithm that is used (with different parameters) in the different phases. In Section~\ref{sec:application} we state the choices of parameters and give additional details about the procedure adopted in the phases of the selection procedure.


\section{Phases of the Selection Procedure}\label{sec:phases}
The ELLIS PhD selection procedure consists of the following phases:
\begin{itemize}
\item \emph{Phase 1 (Student-Evaluator Matching for Pre-Screening)}:
Students state up to five research fields and optionally rank up to 10 professors. Each professor provides up to five research fields and specifies an evaluator who pre-screens applications (e.g. a postdoc).

Based on research fields and on the ranking of professors made by students, both students and evaluators are assigned preferences over the other side of the market. A matching algorithm runs to match students and evaluators. Each student is matched to three evaluators.
Each evaluator is assigned the same maximum capacity of students to pre-screen.
Evaluators score the assigned students.

\item \emph{Phase 2 (Student-Advisor Matching for Interviews)}:
After filtering out badly scoring students, professors rank students based on the pre-screening scores and other means.
Each professor ranks a number of students from 1 - 4 (best to worse) reflecting their priorities for interviews. 5 and 6 indicate that the scoring advisor does not consider the applicant for interviews. 5 specifies that the student is not a good fit for the scoring advisor, but may be good for someone else; 6 marks the student as unfit.
An advisor may be assigned to interview any of their ranked students. Each professor specifies their (maximum) interviewing capacity.

Using preferences of students from \emph{Phase 1}, the ranking assigned by advisors to students in this phase, plus the overlap of research fields, the matching for interviews is computed. Each student must be interviewed by three advisors (ELLIS Excellence Criterion).\footnote{
Due to the limited interviewing capacity of advisors and the high number of ranked students, it proved impossible to match all promising candidates to 3 interviews. In practice, then, students are also matched for 1 or 2 interviews.
}
The algorithm provides each advisor with a number of students to interview not greater than their interviewing capacity.
Advisors must interview the first 80\% of the assigned students, the remaining 20\% are provided on a voluntary basis.

\item \emph{Phase 3 (Student-Advisor Matching for Hiring)}:
Advisors and students rank the other side again.
Advisors also provide their (maximum) hiring capacity.
Based on their updated preferences, advisors and students are matched together. Each student is matched to at most one advisor and each advisor is matched up to their specified hiring capacity.

\item \emph{Phase 4 (Student-Co-advisor Matching for Hiring)}:
Advisors send out acceptance letters to all the students they got matched with. They make sure that they find a co-advisor for any student who accepts.
This year, there is no algorithmic support to find co-advisors.

\item \emph{Phase 5 (Matching the Unmatched)}:
Students who are ranked for offers, but not matched, will enter the pool for the rematching phase. Advisors who were not matched or still have extra hiring capacity can see in the system which students are still unmatched but positively scored. This year, there is no algorithmic support for the rematching stage.

\end{itemize}

\section{Matching theory -- polygamous market}\label{sec:polygamous}
This section explains the theory behind the matching algorithm (Algorithm~\ref{alg:projectAlgorithm}) which is used in each of the phases of the ELLIS PhD program selection. Our work is based on an extensive literature \citep{knuth1997stable,gale1962college,roth1992two} which culminated in the 2012 Economic Nobel Prize to Alvin E. Roth and Lloyd Shapley for their work on stable allocations and the practice of market design\footnote{\url{https://www.nobelprize.org/nobel_prizes/economics/laureates/2012/}}.
Two-sided matching markets are game-theoretic abstractions which correspond to bipartite matchings. We recall some basic results for marriage markets in Appendix~\ref{app:mm} and their extensions to the college admission problem in Appendix~\ref{app:cap}. The theory below provides the theoretical background for the case of two-sided markets with quotas on both sides, also known as many-to-many markets.

Consider the setting where students need to be matched to advisors for interviews as in the case at hand.
Each student can take a maximum number of interviews and each professor can interview a limited number of students. Furthermore, no student should be assigned to the same advisor more than once. 
Students and advisors form the two sides of a matching market. 
We can phrase this problem as follows.

\textbf{We adopt the convention with men and women for consistency with the literature, and we want to stay out of the gender debate. We use the pronouns ``his/her'' to make clear that we refer to a man/woman in the model, so ``person or himself'' can also mean ``person or herself''.}
Let $M$ and $W$ be finite sets of men and women, and let each person be endowed with a strict order of preference with respect to the members of the opposite sex. For each person $p \in W \cup M,$ define $q_{p}$ to be the ``quota''/``capacity''
of this person, i.e., the amount of spouses from the opposite sex this person seeks. 
We call this market a \emph{``polygamous market''}, in reference to the classical marriage market problem introduced by \cite{gale1962college}. In this case, we want to allow quotas on both sides, unlike college admission markets and marriage markets in Appendix~\ref{app:mm}~and~\ref{app:cap}. 

Let us formally define preference lists and the polygamous market.

\begin{definition}[preference lists]
For a market over men $M$ and women $W,$ each man $m$ has preferences $P(m)$ over all women in $W \cup \{ m \}$ defined by the binary relations $\geq_m, =_m$ (which defines $>_m, <_m, \leq_m$). A man has strict preferences if $=_m$ is equal to $=$, i.e. if he is not indifferent between any two women.
Analogously, each woman $w$ has preferences $P(w)$ over all persons in $M \cup \{ w \}$. 
Person $p_1$ is acceptable to $p$ if $p_1 >_p p$.
For persons $p, p_1$ and a set of persons $K$, we define $p_1 >_p K \iff \exists \tau \in K: p_1 >_p \tau$.
\end{definition}

We assume that the spots are independent such that individual preferences suffice to express the preferences over assignments of up to $q_p$ persons.
\begin{definition}[polygamous market]
A polygamous market over men $M$ and women $W$ is defined by the quadruple $(M, W, P, Q)$, where:
\begin{itemize}
    \item $Q = \{ q_m \mid m \in M \} \cup \{ q_w \mid w \in W \}$ are the quotas of men and women,
    \item $P$ is the set of preference lists of men and women:
    \[P=\left\{P\left(m_{1}\right), \dots, P\left(m_{|M|}\right), P\left(w_{1}\right), \dots, P\left(w_{|W|}\right)\right\}.\]
\end{itemize}
\end{definition} 

\begin{definition}[valid matching]
A matching on the polygamous market $(M, W, P, Q)$ is a multi-set (i.e. elements can appear more than once) $\mu \subset M \times W$ such that:
\begin{itemize}
    \item $|\mu(w)| \leq q_w \; \forall w \in W$, where $\mu(w) = \{m \mid (m, w) \in \mu \}$,
    \item $|\mu(m)| \leq q_m \; \forall m \in M$, where $\mu(m) = \{w \mid (m, w) \in \mu \}$.
\end{itemize}
\label{def:validMatchingQuotasPoly}
\end{definition}
We have the property $m \in \mu(w) \iff w \in \mu(m)$. 
$|\mu(m)| \leq q_m$ means that $m$ is matched $q_m - |\mu(m)|$ times to itself and we fill $\mu(m)$ with $m$ up to size $q_m$, same for women.
Since a set does not allow duplicate elements, the same man and woman can match at most once.\footnote{
If $\mu$ were a multi-set, they could match several times. In this case, stable matchings can be obtained by running the traditional GS algorithm on the extended market, where both sides are replicated according to their capacities.
}
When the quotas $Q$ satisfy the college admission market assumption, i.e. one side of the market has quotas all equal to one, this definition coincides with the definition given in Appendix~\ref{app:cap}.

A matching $\mu$ is unstable if any two persons prefer to be together rather than with their assigned partners. From this, a new matching can be constructed.
It should be valid in the sense that it matches any pair at most once. So $(m, w)$ is only a blocking pair if it is not matched already in $\mu$.
On marriage and college admission markets, this definition coincides with the definitions in Appendix~\ref{app:cap}.

\begin{definition}[stability]
\label{def:stabilityProject}
A matching $\mu$ is unstable if there exists a pair $(m, w) \in (M \times W) \cup \{ (m, m) \mid m \in M \} \cup \{ (w, w) \mid w \in W \}$ such that $(m = w \lor m \notin \mu(w))$ and $w >_m \mu(m)$ and $m >_w \mu(w)$. $(m, w)$ is called a blocking pair of $\mu$.\footnote{
The case $m = w$ is also known as individual rationality.
}
\end{definition}

In Algorithm \ref{alg:projectAlgorithm}, we introduce \Palg, which is a direct reformulation of the Gale-Shapley (GS) algorithm to the case of quotas on both sides of the polygamous market.
All men are initially unmatched and women start by matching with themselves as many times as they have capacity.
At each step of the algorithm, a man with available quotas proposes to the woman he prefers most among all women who have not (yet) rejected him. Next, the woman compares this offer with the least favourite man among all the ones she is provisionally engaged to.
If the new proposal is worse according to her preference list, she directly rejects it. 
If the new proposal is better according to her preference list, she disengages the least favourite man and provisionally engages with the new one. As long as a man hasn't filled his quotas, he continues proposing.
If no acceptable woman is left, he fills the remaining spots with himself.
Furthermore, we use the notation $[w]^{q_w}$ to refer to the list $[w, \dots, w]$ of size $q_w$.
In Algorithm~\ref{alg:projectAlgorithm}, the function $\ON(P,m)$, returns $m$'s most preferred woman who did not reject him yet, given his preferences $P(m)$.
In the algorithm, we define $\mu(m) = \{w \mid (m, w) \in \mu \}$, so it is not filled to capacity (during the algorithm).
The function $\WM(w, \mu(w), P)$ returns the weakest match $p \in \mu(w)$ of the woman $w$ ($p$ can be either a man or the woman herself) according to her preferences $P(w)$.
Note that $m >_w \WM(w, \mu(w), P)$ is equivalent to $m >_w \mu(w)$. 

\begin{algorithm}[ht]
\DontPrintSemicolon
\SetAlgoLined
\KwData{Market $(M, W, P, Q)$, quotas $Q$}
\KwResult{Matching $\mu$}
$\mu \gets \cup_{w\in W}\cup_{i=1}^{q_w} \{(w,w)\}$ \tcp{\small match every woman $q_w$ times to herself}
\While{there is a man $m$ with available capacity, i.e. $|\mu(m)| < q_m$}{
 $w \gets \ON(P, m)$ \tcp{\small best woman $w$ to which $m$ has not yet proposed, otherwise $m$}
 \uIf{$w = m$}{
   \tcp{\small man proposed to himself and accepts}
   $\mu \gets \mu \cup \{ (m, m) \}$
 }
 \ElseIf{$m >_w m' = \WM(w, \mu(w), P)$}{
   $\mu \gets \mu \setminus \{ (m', w) \}$ \tcp{\small unmatch $(m', w)$ (once only)}
   $\mu \gets \mu \cup \{ (m, w) \}$ \tcp{\small match $(m, w)$}
 }
}
\Return{$\mu$}
\caption{\Palg}
\label{alg:projectAlgorithm}
\end{algorithm}

The algorithm coincides with the GS algorithm on marriage markets and with the college GS algorithm on college admission markets.
We now prove that this algorithm returns a matching which matches the same man and woman at most once. Additionally, it is stable and optimal for the side which is proposing.
The proofs are adaptations of the many-to-one case \cite{gale1962college}.

\begin{proposition}
\Palg{} terminates and returns a valid and stable matching (Definition~\ref{def:stabilityProject}).
\end{proposition}
\begin{proof}
The algorithm terminates because each man proposes to each woman at most once. 
The returned matching $\mu$ satisfies the quotas because $|\mu(w)| = q_w$ is preserved over iterations for all women $w$ and the algorithm terminates only when $|\mu(m)| = q_m$ for all men $m$. 
Since each man proposes to each woman at most once, a man can be matched to a woman at most once. Therefore, the matching is valid.

To prove stability, we use the following property: Over iterations, the weakest match of any woman $w$ cannot decrease.
By contradiction, let $(m, w)$ a blocking pair of the matching $\mu$.
If $m = w := p \in M \cup W$, this implies $p >_p \mu(p)$, but this can never happen because only acceptable partners are matched. 
Otherwise, $m$ is a man and $w$ is a woman and there exist $m', w'$ (not necessarily man and woman) such that $w >_m w' \in \mu(m)$ and $m >_w m' \in \mu(w)$, where $m'$ is the weakest element of $\mu(w)$. Since $w >_m w' \in \mu(m)$, $m$ must have proposed to $w$. Since $m >_w m' \in \mu(w)$ and the weakest match can never decrease, $w$ would never have rejected $m$, which contradicts $m \notin \mu(w)$.
\end{proof}

Looking at the proof of optimality in the college admission market (without passing via the extended market), we see that it can be extended without problems to this setting.
The difference is that a matching $\mu$ is unstable only if the blocking pair $(m, w)$ is not already part of $\mu$, i.e. $m \notin \mu(w)$.

\begin{proposition}[Optimality]\label{prop:optMain}
Assume strict preferences and men propose (with quotas on both sides). \Palg{} returns a man-optimal result $\mu$, i.e. for every man $m$: $\mu(m)_i \geq_m \mu'(m)_i \; \forall m, i$. The assignments $\mu(m), \mu'(m)$ are ordered from best to worst in terms of $\geq_m$ and $\mu'$ is stable according to Definition \ref{def:stabilityProject}.
Under strict preferences, \Palg{} returns a unique result (independently of the order in which men propose).
\end{proposition}
\begin{proof}
Uniqueness follows immediately from optimality.
A woman $w$ is achievable to man $m$ if there exists a stable matching $\mu'$ such that $m \in \mu'(w)$.
It is enough to prove that a man $m$ is never rejected by an achievable woman.
Indeed, assume there exists a man $m$ and let $i$ the first index such that $\mu(m)_i <_m \mu'(m)_i$.
By assumption, $\mu(m)_j \geq_m \mu'(m)_j >_m \mu'(m)_i$ for all $j < i$. Also, $\mu'(m)_i >_m \mu(m)_i \geq_m \mu(m)_j$ for all $j \geq i$ and this means that man $m$ was rejected by $\mu'(m)_i$ since $m$ applied to $\mu(m)_i <_s \mu'(m)_i$ and is not matched to $\mu'(m)_i$.

Let $m$ the first man who is rejected by an achievable woman $w$ during the execution of the algorithm.
Since $w$ is achievable to $m$, let $\mu'$ the stable matching in which $m$ is matched to $w$, i.e. $m \in \mu'(w)$ (and $w \in \mu'(m)$).
Since $m$ was rejected (and is acceptable to $w$), there must be $q_w$ other men who are all preferred by the woman: $m_i >_w m \in \mu'(w) \; \forall i=1, \dots, q_w$. 
Because none of these other men $m_i$ was yet rejected by an achievable woman, $\forall i \, \exists j_i: w \geq_{m_i} \mu'(m_i)_{j_i}$.
Indeed, by contradiction, assume that there exists an index $i$ s.t. $w <_{m_i} \mu'(m_i)_{j} \; \forall j=1,\dots,q_{m_i}$. 
$m_i$ cannot be matched to $w$.
Since $m_i$ applied to $w$, this means that $m_i$ was rejected by the achievable $\mu'(m_i)_{j_i}$ for some $j_i$, which is a contradiction with $m$ being the first man with this property.

Since $m \in \mu'(w)$, there must exist $m_i \notin \mu'(w)$. Thus $w \neq \mu'(m_i)_{j_i}$ and $w >_{m_i} \mu'(m_i)_{j_i} \in \mu'(m_i)$. Thus, $(m_i, w)$ (with the additional property $m_i \notin \mu'(w)$) blocks $\mu'$, which contradicts the stability of $\mu'$. Hence $w$ is not achievable to $m$.
\end{proof}
In fact, one can also prove that it is woman-pessimal, i.e. $\mu'(w)_i \geq_w \mu(w)_i \; \forall w, i$ for all stable matchings $\mu'$.

As for the GS algorithm, women can propose and we obtain a woman-optimal matching.
Optimality generally does not hold when preferences are not strict.
When preferences are non-strict, we can break ties and the algorithm returns a matching that is also stable with respect to the original preferences.

\section{Application to ELLIS 2020}\label{sec:application} 
As part of the selection procedure, ELLIS aims to match people on one side of the market to the other side, possibly several times.
We will always use the \Palg{} with students proposing.
It reduces to the traditional GS algorithm and college admission algorithm when quotas are all equal to one or equal to one on one side respectively.
Students always propose to ensure student-optimality. 
This means that students get the best matches among all the stable ones.
In each of the phases, we need to specify the market participants and their preferences as well as their capacities.
When the preferences provided to the algorithm are non-strict, they are broken arbitrarily. To have more control over the tie-breaking, we break some of the ties beforehand.
We construct preferences based on the similarity score between fields of research. For person $p_1$, the research similarity score with person $p_2$ is
\begin{align*}
S(p_1, p_2) = R(p_1)^T \cdot R(p_2),
\end{align*}
where $R(p)$ denotes the multi-one-hot encoding of the research interests of person $p$. Say $[A, B, C]$ are the available research fields and person $p$ is interested in fields $A$ and $C$, then $R(p) = [1, 0, 1]^T$.
A person can use this score to rank the people on the other side of the market.

We now describe each of the phases.
Incomplete, fake and blatantly bad students are removed from the system before each phase.
Removed students may be added again after manual inspection at each phase, e.g. students who got kicked out because they weren't matched to three evaluators (Phase 1), three interviews (Phase 2) or to an advisor (Phase 3).

\subsection{Phase 1 - Pre-Screening}
By December 1, students and professors specify their areas of interest. Students may additionally rank up to 10 professors.
Each professor provides an evaluator to pre-screen applications. The evaluator's research fields are those of the corresponding professor and the market consists of evaluators and students. 
Evaluators rank students based on research similarity score.
Students give the best ranks to the (up to 10) advisors (evaluators) they listed, followed by all others ordered by research overlap.
More precisely, given a student, let $\mathcal{A}$ the ordered set of (up to 10) advisors who were ranked by the student.
Let $\mathcal{B}$ the ordered set of all advisors ordered by decreasing research overlap score.
Then, the student's preferences become:
\begin{align*}
\mathcal{A} + (\mathcal{B} \setminus \mathcal{A}),
\end{align*}
where the $+$ operation appends the second list to the first list preserving the order.
Since preferences are based on discrete scores, ties can occur. Advisors break ties between any indifferent students such that they prefer students who listed them.
For example, if $\mathcal{A} = \{ a_1, [a_2, a_3] \}, \mathcal{B} = \{ [a_1, a_4], a_2, [a_3, a_5, a_6] \}$, the new preferences are $\{ a_1, [a_2, a_3], a_4, [a_5, a_6] \}$. If advisor $a_1$ has preferences $\{ [s_1, s_2, s_3, s_4], [s_5, s_6, s_7] \}$ and $s_1, s_2, s_5$ are the only students who listed $a_1$, the advisor's preferences become $\{ [s_1, s_2], [s_3, s_4], s_5, [s_6, s_7] \}$.
The remaining ties are broken arbitrarily.
Students all have quota $3$. Each evaluator has quota
$\ceil*{\frac{3 |S|}{|A|}}$, where $|S|, |A|$ are the total number of students and advisors. 
In the period December 2 - 4, the algorithm runs.
Students are assigned to evaluators. 
From December 5 - 10, evaluators score each student (using the scores ``A'', ``A-B'', ``B'', ``B-C'' and ``C'').
Because a student may not get assigned to three evaluators (e.g. in the scenario with only one evaluator), we break ties again randomly up to 10 times.
If this is unfruitful, we remove (a subset of) insufficiently matched students and rerun the algorithm. This is repeated until a solution is found.

In Figure~\ref{fig:phase1RankDistribution}, we plot a graph that shows that a very large proportion of students gets matched to their first choice.  

\paragraph{Bound on the number of removed students:}
Since we remove students, the matching may not be stable with respect to the original preferences over all students (including the removed students).
We can bound the maximum number of removed students.
Let $q_s$ the maximum (and target) capacity of students, $q_s^{\text{min}}$ the minimum number of matches a student must have in order not to be removed.
The evaluator capacity is $q_e = \ceil*{\frac{q_s |S|}{|E|}}$.
Suppose $k$ students were removed and consider the next iteration (one iteration corresponds to tie breaking up to 10 times).
This means that at most $q_s (|S| - k)$ evaluator spots are occupied, i.e. at least $q_e |E| - q_s (|S| - k) \geq q_s k$ evaluator spots are free.
In this iteration, a student cannot be matched if all these free spots are distributed over $q_s^{\text{min}} - 1$ evaluators.
Once $q_s k > (q_s^{\text{min}} - 1) \cdot q_e$, every student is guaranteed to find enough evaluators and the algorithm terminates.
This holds when $q_s k > (q_s^{\text{min}} - 1) \cdot (\frac{q_s |S|}{|E|} + 1)$.
Therefore, the maximum number of removed students is $k \leq \floor*{(q_s^{\text{min}} - 1) \cdot (\frac{|S|}{|E|} + \frac{1}{q_s}) + 1}$. The fraction of removed students is $\frac{k}{|S|} \leq (q_s^{\text{min}} - 1) \cdot (\frac{1}{|E|} + \frac{1}{|S| q_s}) + \frac{1}{|S|}$. As expected, $k$ is smaller the smaller the ratio $\frac{|S|}{|E|}$ is.
In Phase 1 of ELLIS, $q_s = 3, q_s^{\text{min}} = 3, |E| \approx 100, |S| \approx 500$, therefore $k \leq 11$ and $\frac{k}{|S|} \leq 2.2 \%$.

\begin{figure}
    \centering
    \includegraphics[height=0.4\textheight]{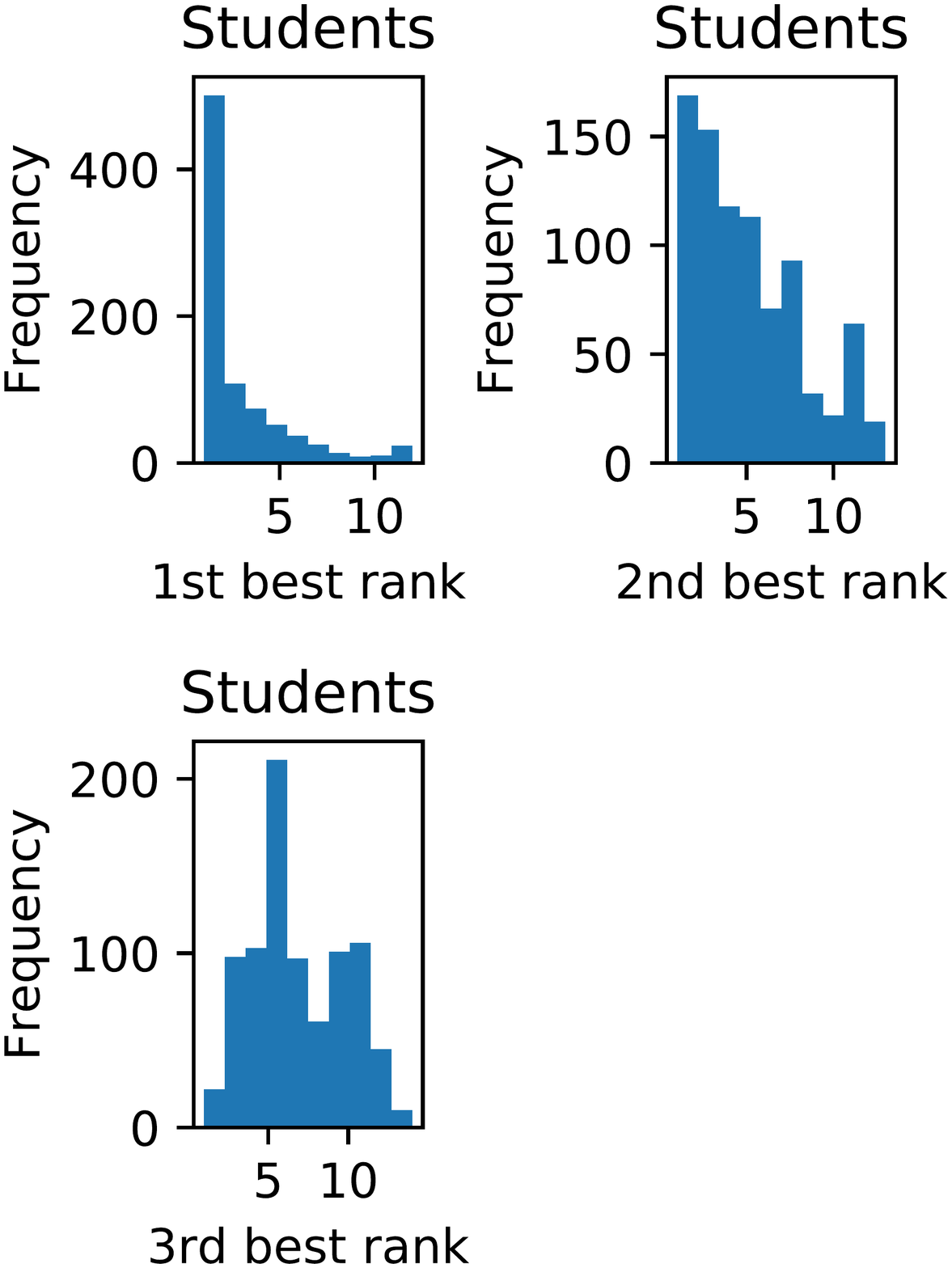}
    \caption{Rank distribution of the students' matches in Phase 1: The top-left shows the number of students whose best match corresponds to their $i$th ranked person (with $i$ on the $x$-axis).
    The top-right is for the second-best match and the bottom-right is for the third-best match.
    }
    \label{fig:phase1RankDistribution}
\end{figure}

\subsection{Phase 2 - Interview Matching}
From December 14 2020 - January 15 2021, professors score students (1 - 4, 5, 6)  based on the application documents, scores from the pre-screening phase and other means\footnote{
In the system, professors can see which students listed them, so they are more likely to rank those.
}.
The score ``5'' means "not a good fit for me, but still good", ``6'' means "this student should not be part of ELLIS" (and is an indication to remove this student from the system). If an advisor ranks a student ``5'' or ``6'', they will not be matched to this student.
Professors also specify their interviewing capacity.
Based on the scores and research overlaps (to break ties of students with same scores), a ranking over scored students is created for each professor. A professor can only be matched to students they scored and didn't score ``5'' or ``6''.
Students are assigned preferences in the same fashion as before. 
Each student has quota $3$ (``ELLIS Excellence Criterion''\footnote{To ensure the excellence of the hired students, this is one of the requirements decided by the ELLIS committee.}). 
While it is certainly possible that an advisor interviews all students that they are interested in, the goal is to also give less ``visible/good'' candidates a chance to get interviewed and possibly hired. Therefore, even the best students can only be assigned three interviews. Additional interviews can be organized individually if necessary.
Each advisor is assigned a quota which is 80\% of their specified interviewing capacity.\footnote{
When the stated interview capacity is less than $3$, it is left as is.
}
We decrease the capacity to 80\% because an advisor must interview all of these 80\% and 100\% may be too severe.
Advisors are not constrained regarding their remaining capacity; we give suggestions for these remaining interviews as outlined below.

Because the ELLIS Excellence Criterion is hard to satisfy, we require each student to be matched to at least 2 interviews rather than 3. If this student should be hired in Phase 3, the criterion can be satisfied by manually arranging an interview.
Therefore, each student has a minimum capacity of 2 and a maximum capacity of 3. When a student is assigned to less than 2 interview slots, they are removed and the algorithm reruns (as described in Phase 1).
Since many students may be matched to less than 2 interviews, we remove at most 20 students at a time (based on average advisor ratings). 
Note that a student with one interview, who was not removed, can match to one interview slot previously taken by a removed student and this avoids their removal.
Finally, when a student has at least one ``1'' and at least one ``5'' or at least two ``1'', they are never removed even if they only have one interview slot.

The algorithm runs in the time frame January 15 - January 22 2021. 
To gain insight into the next-best matches, students and advisors participate again with their remaining quota (including the additional 20\%). No students are removed in this second matching if they have too few interviews. These new matches are included in the optional list of candidates an advisor may wish to interview. However, ELLIS does not put any restrictions on how advisors eventually fill their additional 20\% capacity.
In Figure \ref{fig:phase2RankDistribution}, we plot the distribution of the rank of the $i$th best match for students with $i=1, \dots, 3$, which is the analogous of Figure \ref{fig:phase1RankDistribution}. 

\begin{figure}
    \centering
    \includegraphics[height=0.4\textheight]{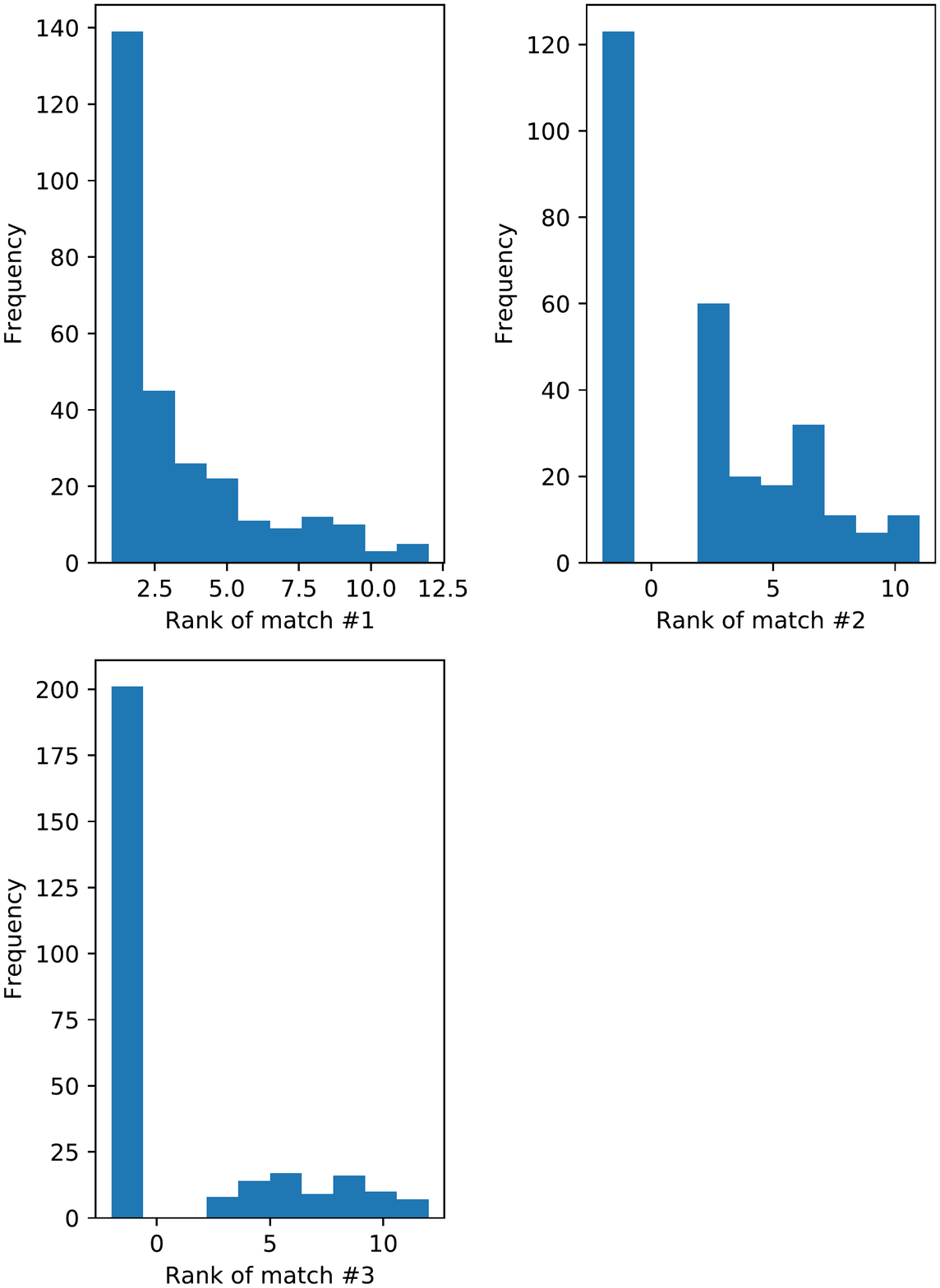}
    \caption{Rank distribution of the students' matches in Phase 2: The top-left shows the number of students whose best match corresponds to their $i$th most preferred person (with $i$ on the $x$-axis).
    The top-right is for the second-best match and the bottom-right is for the third-best match.
    If a person is matched less than three times, the unmatched spots are assigned rank -1 (top-right and bottom-left).
    }
    \label{fig:phase2RankDistribution}
\end{figure}

\subsection{Phase 3 - Advisor Matching}
Between January 23 - February 19 2021, students and professors arrange interviews. 
Professors specify their hiring capacity.
Students and advisors update their preferences by February 19. 
Advisors can be indifferent (e.g. score two students ``1''), but students cannot (and rank up to 10 advisors).
For each advisor, we break ties between equally-scored students based on the total number of scores different from 6 a student got, and randomly break any remaining ties.
These preferences are input unmodified to the matching algorithm (and not modified as before based on the overlap of research fields).
Students have quota 1 and advisors have quota equal to their hiring capacity.
The algorithm is run. Since the number of matches can vary depending on how ties are broken, the algorithm was rerun 10 times and the matching with the maximum number of matches was taken. This resulted in 1 or 2 additional matches.
Around March, advisors send out letters of acceptance to all students they matched with. 
In Figure \ref{fig:phase3RankDistribution}, we plot the distribution of the rank of the $i$th best match for students and advisors respectively, which is the analogous of Figure \ref{fig:phase2RankDistribution}.
All students entering Phase 3 were matched. Most students were matched to their first choice, indicating that people's preferences crystallized out thanks to the previous phases. The matching algorithm was helpful with assigning the remaining persons.
Moreover, we see that all participating advisors were matched at least once. Some advisors were only matched once, either because they had only one spot or the students they were interested in found different matches.

\begin{figure}
    \begin{subfigure}{.5\textwidth}
      \centering
      \includegraphics[height=0.22\textheight]{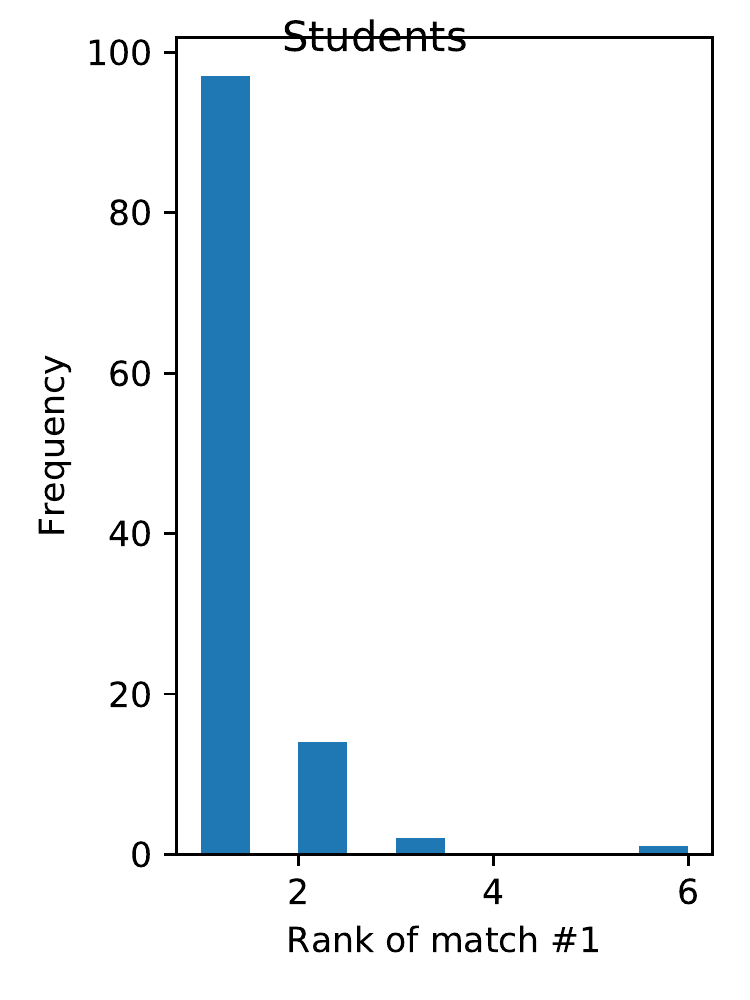} 
    \end{subfigure}%
    \begin{subfigure}{.5\textwidth}
      \centering
      \includegraphics[height=0.4\textheight]{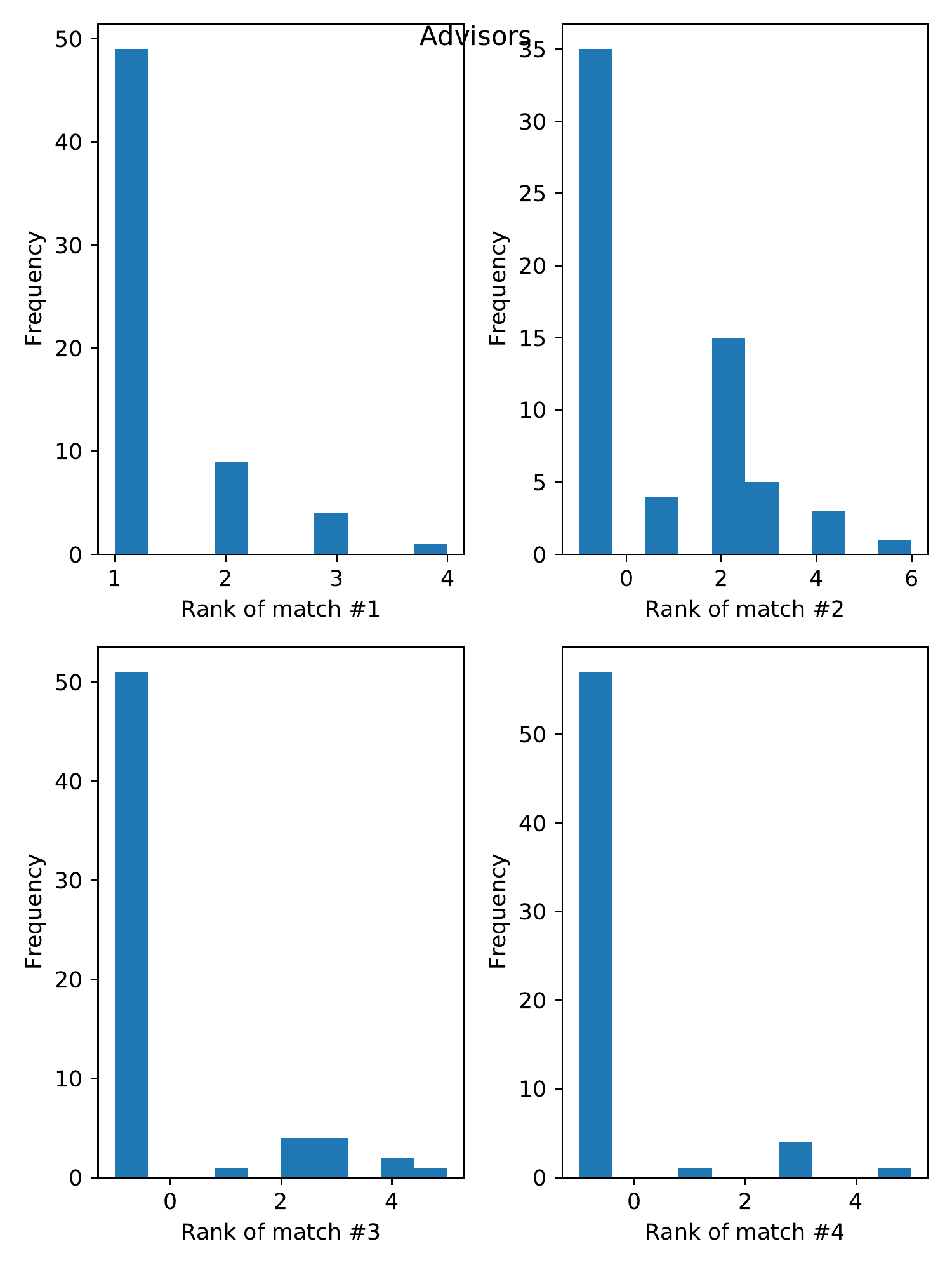}
    \end{subfigure}%
    \caption{Rank distribution of the matches in Phase 3 for students on the left and advisors on the right.
    Rank of match \#i plots a histogram over all persons of the rank each person assigns to their $i$-th best match.
    If an advisor is not matched up to rank \#i, the histogram counts it as rank -1. For example, all advisors were matched at their first spot, whereas 35 advisors were not matched at their second spot (either because they had hiring capacity one or could not find a student).
    }
    \label{fig:phase3RankDistribution}
\end{figure}

\subsection{Phase 4 - Co-Advisor Matching}
This phase is manual.
When a student accepts, the advisor has to find a co-advisor from a different country as required for admission to the ELLIS PhD Program. Students and advisors may already discuss potential co-supervisors during the interview stage.



\section{Acknowledgements}
We want to thank the following persons for the initial idea and help with the practical aspects of this project: Andreas Geiger, Lynn Anthonissen, Leila Masri, and the other members of the ELLIS PhD Committee.

\bibliographystyle{plainnat} 
\bibliography{references}

\appendix

\section{Marriage-market}\label{app:mm}

In the seminal work by \cite{gale1962college}, the authors introduced two types of matching markets, the marriage markets that we are going to recall in this section, and the college admission markets (Section~\ref{app:cap}). These notions are at the basis of the extension in Section~\ref{sec:polygamous}.
In a classical marriage market consisting of a set of men $M$, women $W$ and preferences of each person over the persons of the other side, the goal is to match each man to at most one woman. Each woman can get married to one man at most. In fact, a person may also choose to match with themself rather than match with some unacceptable partner.
We start from some formal definitions and these are illustrated in Example \ref{ex:marriageMarket}.

\begin{definition}[preference lists]
For a market over men $M$ and women $W,$ each man $m$ has preferences $P(m)$ over all persons in $W \cup \{ m \}$ defined by the binary relations $\geq_m, =_m$
(which defines $>_m, <_m, \leq_m$). A man has strict preferences if $=_m$ is equal to $=$, i.e. he is not indifferent between any two people.
Analogously, each woman $w$ has preferences $P(w)$ over all persons in $M \cup \{ w \}$. 
Person $p_1$ is acceptable to $p$ if $p_1 >_p p$.
\end{definition}

These preferences can be represented as ordered lists as shown in Example \ref{ex:marriageMarket}. 

\begin{definition}[marriage market] 
A marriage market over men $M$ and women $W$ is defined by the triple $(M, W, P)$, where $P$ is the set of preference lists, $P=\left\{P\left(m_{1}\right), \dots, P\left(m_{n}\right), P\left(w_{1}\right), \dots, P\left(w_{p}\right)\right\}$.
\end{definition}

A matching $\mu$ is valid if each person is matched to exactly one partner from the opposite sex or themself.

\begin{definition}[valid matching]
A matching on the marriage market $(M, W, P)$ is a correspondence $\mu: M \cup W \rightarrow M \cup W$ such that:
\begin{itemize}
    \item $\mu(m) \in W \cup\{m\},$
    \item $\mu(w) \in M \cup\{w\},$
    \item $\mu(m) = w \iff m = \mu(w).$
\end{itemize}
\end{definition}

While there exist many valid matchings, matchings should not fall apart quickly because people find better partners, thus disregarding the matching. 
This can be ensured if the matching is stable.
A matching is stable if there does not exist a man $m$ and woman $w$ such that $(m, w)$ strictly prefer each other to the partners they are currently matched with. 
In addition, each person prefers to stay single rather than match with some unacceptable partner.
This leads to the following definition.

\begin{definition}[stability]
A matching $\mu$ is unstable if there exists a pair $(m, w) \in (M \times W) \cup \{ (m, m) \mid m \in M \} \cup \{ (w, w) \mid w \in W \}$ such that $w >_m \mu(m)$ and $m >_w \mu(w)$. $(m, w)$ is called a blocking pair of $\mu$\footnote{The case $m=w$ is also known as individual rationality in the literature.}.
\end{definition}

Stability implies that it is enough to list all acceptable partners up to the position of the person themself. A person will always prefer to match to themself rather than match to anyone coming afterwards in their preferences. 

At this point, let us consider an example.

\begin{example}
\label{ex:marriageMarket}
Consider the market with men $M = \{ m_1, m_2 \}$ and women $W = \{ w_1, w_2, w_3 \}$ and preferences $P$ (expressed as a list in the order of decreasing preference):
\begin{align*}
P(m_1) &= \{ w_1, w_2, w_3, m_1 \},\quad
P(m_2) = \{ w_2, w_1, m_2 \},\\
P(w_1) &= \{ m_2, m_1, w_1 \},\quad
P(w_2) = \{ m_1, m_2, w_2 \},\quad
P(w_3) = \{ m_1, w_3 \}.
\end{align*}
We could equivalently write $P(m_2) = \{ w_2, w_1, m_2, w_3 \}$ with woman $w_3$ unacceptable to $m_2$, but this does not affect the set of stable matchings because $m_2$ will always prefer himself to $w_3$.
One can verify that the matching $\mu = \{ (m_1, w_2), (m_2, w_1), (w_3, w_3) \}$ is stable with woman $w_3$ matched to herself.
Another stable matching is $\mu = \{ (m_1, w_1), (m_2, w_2), (w_3, w_3) \}$.
Indifferent preferences are represented by square brackets:
\begin{align*}
P(m_1) = \{ [w_1, w_2], w_3, [w_4, w_5], m_1 \}.
\end{align*}
It means that $w_1 =_m w_2 >_m w_3 >_m w_4 =_m w_5 >_m m_1 >_m \textrm{anyone else}$.
A person can never be indifferent between themself and anyone else, i.e. $[w_5, m_1]$ is not allowed in the preferences of $m_1$.
\end{example}

Though desirable, it is questionable whether a stable matchings can be found in general.
The celebrated Gale-Shapley algorithm (\GSalg) presented in \citep{gale1962college}, shows by construction that a stable matching exists in all marriage markets. 
The pseudocode of the \GSalg{} is provided in Algorithm \ref{alg:galeShapley}. The \GSalg{} is typically exposed by letting men propose at the same time. 
The version presented here makes the generalization in Section~\ref{sec:polygamous} more straightforward.

\begin{algorithm}[ht]
\DontPrintSemicolon
\SetAlgoLined
\KwData{Marriage market $(M, W, P)$}
\KwResult{Matching $\mu_M$}
$\mu \gets \{ (w, w) \mid w \in W\}$ \tcp{match every woman to herself}
\While{there is an unmatched man $m$}{
 $w \gets \ON(P, m)$ \tcp{partner: woman or man himself}
 \uIf{$w = m$}{
   $\mu \gets \mu \cup \{ (m, m) \}$ \tcp{man proposed to himself and accepts}
 }
 \ElseIf{$m >_w \mu(w)$}{
   $\mu \gets \mu \setminus \{ (m', w) \}$ \tcp{\small unmatch $(m', w)$}
   $\mu \gets \mu \cup \{ (m, w) \}$ \tcp{\small match $(m, w)$}
 }
}
\Return{$\mu$}
\caption{Deferred Acceptance Algorithm or Gale-Shapley Algorithm}
\label{alg:galeShapley}
\end{algorithm}

The GS algorithm works by letting men propose to women and women conditionally accept unless they get an offer from a better man later on. It starts with all men unmatched and all women matched to themselves.
As long as a man is unmatched, consider any unmatched man. He proposes to his next most preferred woman he has not proposed to already. The function that does this in Algorithm~\ref{alg:galeShapley} is $\ON$.
In case of indifferent preferences, a man arbitrarily picks any of the equally preferred women.
If a man has proposed to all of his acceptable women, he proposes to himself instead (he accepts and remains single).
If the woman prefers the man to her current partner, she disengages from her old partner (a man or herself), leaving her old partner unmatched again. She engages/matches with this new man.
The algorithm stops once all men are matched (with a woman or themeselves).
The matched men and women are married.
The algorithm is also termed \emph{deferred acceptance algorithm} since a woman confirms her engagement only once the algorithm terminates and may break up for a better man any time before.
The algorithm does not specify the order in which free men are chosen or how a man decides between indifferent women. When preferences are strict, the returned matching is always the same.

\begin{theorem}[Gale and Shapley]
The GS algorithm terminates and returns a valid and stable matching.
\end{theorem}
\begin{proof}
The algorithm terminates because the preference lists of men are finite and a man always accepts himself. 
The returned matching is valid because a man is matched to exactly one partner, himself or a woman.

By contradiction, we prove that the matching is stable.
Assume $(m, w)$ blocks $\mu$. If $m = w$, this means that $m$ is either a man or a woman and matched to an unacceptable partner in $\mu$. This is impossible because only acceptable partners are matched in the algorithm. Otherwise, $m$ is a man and $w$ a woman and $w >_{m} \mu(m)$, $m >_w \mu(w)$. This means that $m$ proposed to $w$ before proposing to $\mu(m)$ and was rejected. Since $w$'s match cannot decrease over iterations, this means that $w$ would have accepted.
\end{proof}

Assuming that an unmatched man can be identified in $O(1)$, the running time is $O(M W)$ since each man proposes to a woman at most once. Whilst it is possible to find an unmatched man in $O(1)$ by storing the free men in a set, uniformly picking a free man is not $O(1)$. One could obtain a random sample from a set in $O(\text{\# free men})$ by passing via a list.



Under strict preferences, it can be shown that the GS matching is optimal in the sense that each man is matched to the best partner he could get among all stable matchings. In addition, the matching is worst-optimal for women, meaning that each woman gets the worst partner among all stable matchings.
\begin{proposition}[Optimality]
Assume strict preferences and let $\mu$ the matching returned by the GS algorithm. Then $\mu(m) \geq_m \mu'(m)$ for each man $m$ and any stable matching $\mu'$. Also, $\mu'(w) \geq_w \mu(w)$ for each woman $w$ and any stable matching $\mu'$.
\end{proposition}
The proof is a special case of the proof in Proposition \ref{prop:optMain}.
Since optimality implies uniqueness, the GS algorithm returns a unique matching under strict preferences, independently of the order in which men propose.
By inverting the roles of men and women, an equivalent version of the GS algorithm lets women propose to men and the matching is generally different. Under strict preferences, this matching is woman-optimal and men-worst-optimal.

\section{College Admission Problem}\label{app:cap}
The marriage market can be generalized to the college admission setting. Instead of men and women, the market consists of colleges and students. A student can go to at most one college, but a college can accept more than one student, up to its capacity.
We will state more general definitions respect to the previous section and some of them can also be applied to the extension in Section~\ref{sec:polygamous} when students can have capacities as well (polygamous market).

\begin{definition}[college admission market]
A college admission market over students $S$ and colleges $C$ is defined by the quadruple $(S, C, P, Q)$, where:
\begin{itemize}
    \item $Q = \{ q_c \mid c \in C \} \cup \{ q_s \mid s \in S \}$ are the capacities of colleges and students,
    \item $P$ is the set of preference lists of colleges and students:
    \[P=\left\{P\left(c_{1}\right), \dots, P\left(c_{n}\right), P\left(s_{1}\right), \dots, P\left(s_{p}\right) \right\}.\]
\end{itemize}
We say that the quotas $Q$ satisfy the college admission market assumption if one side of the market has quotas all equal to one. Without loss of generality, we call students the side with quotas all equal to one.
\end{definition}
The college admission market is the market where only one side has quotas greater than one, i.e. $q_s = 1 \; \forall s \in S$.

\begin{definition}[valid matching]
A matching on the college admission market $(S, C, P, Q)$ is a set $\mu \subset S \times C$ such that:
\begin{itemize}
    \item $|\mu(c)| \leq q_c \; \forall c \in C$, where $\mu(c) = \{s \mid (s, c) \in \mu \}$,
    \item $|\mu(s)| \leq q_s (=1) \; \forall s \in S$, where $\mu(s) = \{c \mid (s, c) \in \mu \}$.
\end{itemize}
$\mu(s)$ and $\mu(c)$ equivalently characterize the matching with the property: $s \in \mu(c) \iff c \in \mu(s)$. $|\mu(c)| \leq q_c$ means that $c$ is matched $q_c - |\mu(c)|$ times to itself and we fill $\mu(c)$ with $c$ up to size $q_c$, same for students.
\label{def:validMatchingQuotas}
\end{definition}

We give an example of a college admission market.
\begin{example}
Consider the college admission market with students $S = \{ s_1, s_2, s_3 \}$ and colleges $C = \{ c_1, c_2 \}$ with capacities $2$ and $3$. Preferences for students over colleges are expressed as before. For colleges, they express their preferences over groups of students of size less or equal to their capacity.
For some preferences, a valid matching could be $\mu = \{ (c_1, \{ s_1, s_2 \}), (c_2, \{ \}), (s_3, s_3) \}$. This can be equivalently written by filling spots to capacity, $\mu = \{ (c_1, \{ s_1, s_2 \}), (c_2, \{ c_2, c_2, c_2 \}), (s_3, s_3) \}$, or by listing the set $\mu = \{ (c_1, s_1), (c_1, s_2) \}$.
\end{example}

We define preferences between a single person and a set of persons as follows.
\begin{definition}
Given a set of persons $K$, we say that $p_1 >_p K$ if there exists $p_2 \in K$ such that $p_1 >_p p_2$.
\end{definition}

The stability definition carries over and we restate it here for clarity. It assumes that $\mu(s)$ and $\mu(c)$ are filled up to their capacity, as described in Definition \ref{def:validMatchingQuotas}.

\begin{definition}[stability]
A matching $\mu$ is unstable if there exists a pair $(s, c) \in (S \times C) \cup \{ (s, s) \mid s \in S \} \cup \{ (c, c) \mid c \in C \}$ such that $s >_c \mu(c)$ and $c >_s \mu(s)$, i.e. there exist $\tau_c \in \mu(c), \tau_s \in \mu(s)$ such that $s >_c \tau_c$ and $c >_s \tau_s$.
In this case, $(s, c)$ is called a blocking pair of $\mu$.
\label{def:stabilityCollege}
\end{definition}

\begin{theorem}
The college admission market admits a valid and stable matching.
\end{theorem}
\begin{proof}[Proof (sketch)]
It is possible to find a stable matching by relying on the GS algorithm. This is illustrated in Algorithm \ref{alg:collegeAlgorithm}.
It works by mapping the college admission market to an extended marriage market, where college $c$ is replicated according to its capacity to $c_1, \dots, c_{q_c}$, each with the same preferences over students as in the original market (with $c$ replaced by $c_i$). Students equally prefer any of the replicated colleges, i.e. any occurrence of $c$ is replaced by the indifferent $[c_1, \dots, c_{q_c}]$, see Example \ref{ex:extendedMarket}.
It is easy to show that a matching is stable on the college admission market if and only if it is stable on the extended marriage market. 
Therefore, the GS algorithm can be used to find a stable matching on the extended market and then map it back.
\end{proof}

\begin{algorithm}[ht]
\DontPrintSemicolon
\SetAlgoLined
\KwData{Market $(S, C, P, Q)$, quotas $Q$ for one side}
\KwResult{Matching $\mu$}
$\tilde{S}, \tilde{C}, \tilde{P}, \textrm{mapping} \gets \TO(S, C, P, Q)$ \;
$\tilde{\mu} \gets \textsc{Gale-Shapley}(\tilde{S}, \tilde{C}, \tilde{P})$ \;
$\mu \gets \From(\tilde{\mu}, \textrm{mapping})$ \;
\Return{$\mu$}
\caption{College Admission Algorithm}
\label{alg:collegeAlgorithm}
\end{algorithm}
In the GS algorithm, either $\tilde{S}$ can propose to $\tilde{C}$ or vice versa.
Traditionally, quotas are such that only one side has quotas greater than one, but the algorithm continues to work for quotas on both sides if we define $\mu$ to be a multi-set so that the same pair can match more than once.
In Section~\ref{sec:polygamous}, we consider the case when both sides have quotas, but the same pair can match at most once.

\begin{example}[extended market construction]
\label{ex:extendedMarket}
Consider students $S = \{ s, \tilde{s} \}$ with capacities 2, 1 and colleges $C = \{ c, \tilde{c} \}$ with capacities 1, 2 and preferences (omitting the person themself in the preferences):
\begin{align*}
P(s) &= \{ c, \tilde{c} \},\quad
P(\tilde{s}) = \{ \tilde{c} \},\\
P(c) &= \{ [\tilde{s}, s] \},\quad
P(\tilde{c}) = \{ s \},\quad
\end{align*}
The extended market maps $s \mapsto [s_1, s_2], \tilde{s} \mapsto [\tilde{s}_1], c \mapsto [c_1], \tilde{c} \mapsto [\tilde{c}_1, \tilde{c}_2]$. It is defined over students $S^{ext} = \{ s_1, s_2, \tilde{s}_1 \}$ and colleges $C^{ext} = \{ c_1, \tilde{c}_1, \tilde{c}_2 \}$ with preferences:
\begin{align*}
P(s_1) &= \{ [c_1], [\tilde{c}_1, \tilde{c}_2] \},\quad
P(s_2) = P(s_1),\quad
P(\tilde{s}_1) = \{ [\tilde{c}_1, \tilde{c}_2] \},\\
P(c_1) &= \{ [\tilde{s}_1, s_1, s_2] \},\quad
P(\tilde{c}_1) = \{ [s_1, s_2] \},\quad
P(\tilde{c}_2) = P(\tilde{c}_1),
\end{align*}
More precisely, $P(s_2) = P(s_1)$ means $P(s_1) = \{ [c_1], [\tilde{c}_1, \tilde{c}_2], s_1 \}$ and $P(s_2) = \{ [c_1], [\tilde{c}_1, \tilde{c}_2], s_2 \}$. Observe what happens to the indifferent preferences of $c$.
The GS algorithm runs on this extended market and the obtained matching is transformed back to give a stable matching on the original market.
\end{example}

Optimality carries over to this setting. In particular, the college-optimal matching is student-worst-optimal and vice versa. This proof can be extended easily to many-to-many markets. It can also be proven by passing via the extended market.
When the same pair can match at most once, this continues to be true because the set of stable matchings that match the same pair at most once is a subset of the set of stable matchings that may match the same pair more than once. The former inherits the ordering of the latter.

\end{document}